\documentclass[11pt]{article}

\usepackage[margin=1in]{geometry}
\usepackage{amsmath,amssymb,amsthm,mathtools}
\usepackage{aliascnt}
\usepackage{enumitem}
\usepackage{microtype}
\usepackage{xspace}
\usepackage{xcolor}
\usepackage[backend=biber,style=alphabetic,maxbibnames=99,minalphanames=3]{biblatex}
\usepackage{hyperref}
\usepackage[nameinlink,capitalise,noabbrev]{cleveref}

\AtBeginBibliography{\small}

\definecolor{PineGreen}{RGB}{34, 92, 70}
\definecolor{OliveCite}{RGB}{85, 107, 47}
\definecolor{WarmAmber}{RGB}{180, 95, 20}

\hypersetup{
    colorlinks=true,
    linkcolor=PineGreen,
    citecolor=OliveCite,
    urlcolor=WarmAmber,
}

\newtheorem{theorem}{Theorem}[section]
\newaliascnt{lemma}{theorem}
\newtheorem{lemma}[lemma]{Lemma}
\aliascntresetthe{lemma}
\newaliascnt{corollary}{theorem}
\newtheorem{corollary}[corollary]{Corollary}
\aliascntresetthe{corollary}
\newaliascnt{proposition}{theorem}

\aliascntresetthe{proposition}
\theoremstyle{definition}
\newaliascnt{definition}{theorem}
\newtheorem{definition}[definition]{Definition}
\aliascntresetthe{definition}
\newaliascnt{construction}{theorem}
\newtheorem{construction}[construction]{Construction}
\aliascntresetthe{construction}
\theoremstyle{remark}
\newaliascnt{remark}{theorem}

\aliascntresetthe{remark}
\crefname{theorem}{theorem}{theorems}
\Crefname{theorem}{Theorem}{Theorems}
\crefname{construction}{construction}{constructions}
\Crefname{construction}{Construction}{Constructions}
\crefname{definition}{definition}{definitions}
\Crefname{definition}{Definition}{Definitions}
\crefname{lemma}{lemma}{lemmas}
\Crefname{lemma}{Lemma}{Lemmas}
\crefname{proposition}{proposition}{propositions}
\Crefname{proposition}{Proposition}{Propositions}
\crefname{corollary}{corollary}{corollaries}
\Crefname{corollary}{Corollary}{Corollaries}
\crefname{remark}{remark}{remarks}
\Crefname{remark}{Remark}{Remarks}

\newcommand{\PPAD}{\mathsf{PPAD}}
\newcommand{\spend}{\operatorname{spend}}
\newcommand{\dec}{\operatorname{dec}}
\newcommand{\NOT}{\textnormal{\textsf{NOT}}\xspace}
\newcommand{\NOR}{\textnormal{\textsf{NOR}}\xspace}
\newcommand{\NPURIFY}{\textnormal{\textsf{NPURIFY}}\xspace}
\newcommand{\PURIFY}{\textnormal{\textsf{PURIFY}}\xspace}
\newcommand{\PureCircuit}{\textnormal{\textsc{Pure-Circuit}}\xspace}
\newcommand{\goodsupport}{\Delta_{\mathrm{good}}^{+}}
\newcommand{\biddersupport}{\Delta_{\mathrm{bidder}}^{+}}

\title{Tight Inapproximability of Pacing and Throttling Equilibria\\ in Second-Price Auctions}
\author{Zhengyang Liu\thanks{Beijing Institute of Technology, zhengyang@bit.edu.cn}}
\date{}

\begin{document}
\maketitle
\begin{abstract}
  Budget-constrained advertisers commonly rely on two control mechanisms: pacing scales bids, whereas throttling randomizes participation. We prove that, in second-price auctions, these two different mechanisms share the same sharp approximation-hardness threshold. For pacing, computing a $\gamma$-approximate equilibrium is $\PPAD$-hard for every constant $\gamma\in[0,1)$. For throttling, computing a $\delta$-approximate equilibrium is $\PPAD$-hard for every constant $\delta\in(0,1)$. At parameter $1$, the complementarity requirement becomes vacuous and the all-zero solution is feasible. That is, approximation does not eliminate the fixed-point barrier at any nontrivial parameter value.
\end{abstract}

\section{Introduction}
\label{sec:introduction}

Advertisers who participate in many auctions must respect a single budget across all of them. Two standard control mechanisms enforce this constraint. \emph{Pacing} multiplies all of a bidder's bids by one common factor, whereas \emph{throttling} leaves bids unchanged and randomizes whether the bidder participates. Despite this operational difference, the two mechanisms implement the same complementarity principle: a bidder should be restricted only when her budget is close to binding.

Whether this principle is computationally benign depends crucially on the payment rule. Under first-price payments, pacing equilibria admit a convex-program characterization~\cite{CKP22}, while throttling has a unique equilibrium computable by t\^{a}tonnement-style decentralized dynamics~\cite{CKK21T}. Under second-price payments, by contrast, both equilibrium problems are $\PPAD$-hard~\cite{CKK24,CL25,CKK21T}. This common change in complexity leads to a natural question: even if exact computation is hard, does sufficiently coarse approximation restore tractability?

Previous work offered partial negative answers. Chen and Li~\cite{CL25} proved pacing hardness for $\gamma < 1/3$. For throttling, the reduction of Chen, Kroer, and Kumar~\cite{CKK21T} from approximate threshold games---combined with the tight $\varepsilon < 1/6$ threshold-game hardness of Deligkas et al.~\cite{DFHM24}---yields hardness for $\delta < 1/19$.\footnote{In the relevant regime, the reduction of Chen, Kroer, and Kumar uses $\delta=\varepsilon/(3+\varepsilon)$; hence $\varepsilon<1/6$ corresponds exactly to $\delta<1/19$.} In this paper, we completely settle the positive-error gap: hardness persists up to the vacuous endpoint for both models.

\paragraph{Main Results.}
For pacing, computing a $\gamma$-approximate equilibrium is $\PPAD$-hard for every constant $\gamma\in[0,1)$. For throttling, computing a $\delta$-approximate equilibrium is $\PPAD$-hard for every constant $\delta\in(0,1)$. At parameter $1$, complementarity disappears and the all-zero solution is feasible. Thus, in both models, hardness persists throughout the entire non-vacuous positive-error range; the exact throttling endpoint is discussed separately in \Cref{sec:discussion}.

Although the two reductions reach the same endpoint, they do so with different sparsity guarantees. For pacing, the hard instance $\mathcal I$ satisfies
\[
  \goodsupport(\mathcal I)=2,
  \qquad
  \biddersupport(\mathcal I)=O((1-\gamma)^{-2}).
\]
For throttling, writing $H:=1-\delta$, it satisfies
\[
  \goodsupport(\mathcal I),\ \biddersupport(\mathcal I)
  =O\left(H^{-1}(1+\log(1/H))\right).
\]
Here $\goodsupport$ is the maximum number of positive bidders on a good and $\biddersupport$ is the maximum number of positive goods for a bidder; \Cref{def:support-degrees} gives the formal definition. The two hardness theorems appear in \Cref{thm:pacing-hardness,thm:throttling-hardness}.

\paragraph{Extensions.}
The same picture survives two natural extensions. First, both constructions remain polynomial when the distance to the vacuous endpoint is inverse-polynomial; see \Cref{cor:pacing-inverse-polynomial-gap,cor:throttling-inverse-polynomial-gap}. Second, the hardness survives when the approximation margins are relaxed independently rather than tied to a single parameter. These stronger formulations are stated in \Cref{thm:independent-pacing,thm:independent-throttling}.

\paragraph{Technical Overview.}
We reduce from bounded-degree \PureCircuit. Each circuit node $v$ is represented by one variable bidder, and the bidder's pacing multiplier or throttling probability encodes the state of $v$. The high-level gate logic is simple: a source node sends one scalar signal to each of its successors, and the budget of an output bidder turns the total incoming signal into a threshold test. A sufficiently small signal leaves enough budget slack to force the output into its high state, while a sufficiently large signal makes the high state infeasible and forces the output into its low state. Thus a \NOT gate is implemented by choosing a budget between the low and high signal levels, and the same threshold test implements \NOR after summing two incoming signals. For \NPURIFY, we use two different output budgets. If both outputs were indeterminate, the common predecessor signal would have to lie in two disjoint ambiguity intervals, which is impossible. This common logic is formalized in \Cref{sec:reduction-framework}; the main difficulty in both reductions is therefore not the circuit simulation itself, but how to transmit a sufficiently clean signal under second-price payments when the approximation parameter approaches the vacuous endpoint.

\medskip
\noindent\emph{Pacing: replacing a positive multiplier floor by geometric isolation.}
Write $\beta:=1-\gamma$. Chen and Li~\cite{CL25} obtain constant hardness by first forcing every variable multiplier that participates in the reduction to stay above a fixed positive floor. This is a very useful trick: once the target multiplier is bounded away from zero, one can give the target a sufficiently large value on a signal good, guarantee that it always wins, and make its second-price payment equal to the source multiplier. Their reduction then encodes the two Boolean states by two multipliers separated by a constant gap and yields hardness for $\gamma<1/3$.

This approach is not robust as $\gamma\to1$. The complementarity condition can then exploit only a $\beta$-fraction of the output budget, so the scale at which the reduction distinguishes low and high signals must itself go to zero. In particular, we can no longer rely on keeping every variable multiplier above a constant floor. Once the low state is allowed to approach zero, however, the target of a signal edge need not always dominate the source. A naive signal good therefore becomes two-way: the target may fail to win, and the source may instead incur a second-price payment depending on the target. Such backward payment changes the source's own budget balance and destroys the intended circuit semantics. Controlling this feedback while still transmitting the source multiplier accurately is the main obstacle to pushing pacing hardness to the endpoint.

Our replacement for the fixed-floor trick is a \emph{geometric wire}. For an interaction edge $u\to v$, let $a$ and $b$ denote the source and target multipliers. The wire consists of $K$ goods. The source has value $1/K$ on every level, while the target value on level $\ell$ is $1/(Kt_\ell)$, where the thresholds $t_\ell$ grow geometrically and all satisfy $t_\ell\le L/2$. We choose $L=\Theta(\beta^2)$ and $K=\Theta(\beta^{-2})$. The wire has two complementary properties. First, whenever the target is outside the low decoding region, namely $b>L$, it uniquely wins \emph{every} level and pays $a/K$ on each one. Hence its total payment on the wire is exactly $a$,
so the source multiplier is transmitted without distortion. Second, if the source wins any levels, those levels form a geometric suffix of the wire. The second prices along this suffix decrease geometrically, and their total is less than $4/(3K)$. Since every \PureCircuit node has out-degree at most two, all backward payments of a variable bidder are only $O(1/K)=O(\beta^2)$.

This is the key improvement over the previous construction. We do not force the low state to stay at a constant multiplier. Instead, we let the decoding threshold shrink to $L=\Theta(\beta^2)$ and make the unwanted feedback shrink at the same rate. The output budget can then distinguish an incoming signal of size at most $L$ from one bounded away from zero even though complementarity has only $\Theta(\beta)$ relative slack. Concretely, if the aggregate incoming multiplier is sufficiently small, the bidder spends less than $\beta B$ and is forced to multiplier $1$; if the aggregate signal exceeds $B$, any multiplier above $L$ would make the exact forward payments exceed the budget, so the bidder must lie in the low region. The remaining intermediate case is confined to a short budget-dependent interval. Two carefully separated budgets give disjoint such intervals for the two outputs of \NPURIFY. The delicate point is that the same wire must simultaneously give \emph{exact} forward transmission whenever the target is not low and \emph{uniformly tiny} reverse leakage for every possible source and target multiplier, including intermediate states and ties; the geometric spacing is what makes both properties hold at once.

\medskip
\noindent\emph{Throttling: turning a constant blocker response into an exponential signal gap.}
For throttling, write $H:=1-\delta$. The difficulty is different. A variable bidder creates no backward payment on an outgoing signal good, but an incoming expected payment is necessarily multiplied by the target's own participation probability. Thus the most direct signal from a source $u$ to a target $v$ has the form $p\theta_u\theta_v$. When $\delta$ is close to $1$, even a decoded high target is guaranteed to participate only with probability $H$, and this multiplicative attenuation becomes severe. Simply repeating the same signal good does not help: it multiplies both the low and high cases by the same factor and therefore does not improve their ratio. The earlier threshold-game route of Chen, Kroer, and Kumar~\cite{CKK21T} yields hardness only in a fixed constant range; its analysis does not provide a parameter-dependent amplification that survives as $H\to0$. Reaching every $\delta<1$ therefore requires a signal whose low--high ratio itself improves as $H$ shrinks.

We obtain such a signal with a \emph{blocker amplifier}. Each source $u$ is equipped with $K$ blocker bidders $z_{u,1},\ldots,z_{u,K}$. A control good couples each blocker to the source. We choose the blocker budget so that
\[
    \theta_u\le L \quad\Longrightarrow\quad \theta_{z_{u,q}}\ge H,
    \qquad
    \theta_u\ge H \quad\Longrightarrow\quad \theta_{z_{u,q}}\le H/4,
\]
where $L=H^3/8$. Thus every individual blocker reacts only by a constant-factor change. The amplification occurs on the signal good: all blockers bid above the target, and the target bids above the source. Consequently, conditional on the target participating, it sees the source as the second-price bidder only when \emph{all} blockers are absent. The effective source signal is therefore
\[
    y_u
    =p\theta_u\prod_{q=1}^K(1-\theta_{z_{u,q}}).
\]
A low source gives $y_u\le pL(1-H)^K$, whereas a high source gives $y_u\ge pH(1-H/4)^K$.
The crucial point is the product. The constant-factor difference between one blocker's low- and high-source responses is raised to the $K$th power, producing an exponential separation. Taking
\[
    K=O\left(H^{-1}(1+\log(1/H))\right)
\]
is enough for this exponential gain to dominate all polynomial losses coming from the vanishing decoding scale $L$. In the proof this is summarized by the master inequality $4a_0<HL^2a_1$, where $a_0$ and $a_1$ are the worst-case low and high signal levels.

Once this amplified signal is available, the rest again becomes a budget threshold test. A variable bidder pays only on incoming signal goods, and in fact its spending has the exact form
\[
    S_v(\theta)=\theta_v\sum_{u\in N^-(v)}y_u.
\]
If the aggregate signal is below $HB$, complementarity forces $\theta_v\ge H$; if it exceeds $B/L$, any $\theta_v>L$ violates the budget. An indeterminate unary output therefore pins its predecessor signal to the interval $(B,B/L)$. Choosing two output budgets with disjoint intervals implements \NPURIFY. We reuse the \emph{same} blocker bundle on all edges leaving a source, so the two outputs of an \NPURIFY gate observe exactly the same scalar $y_u$, which is essential for this argument.

The two reductions therefore overcome opposite second-price effects. Pacing must suppress an \emph{additive backward-feedback error}, which we do by geometric isolation; throttling must overcome a \emph{multiplicative attenuation}, which we do by exponential blocker amplification. In both cases the approximation parameter approaching $1$ makes the available complementarity signal vanish, so a constant-gap gadget is insufficient. The technical contribution is to make the mechanism-specific transmission error vanish faster than the budget test degenerates: $O(\beta^2)$ leakage for pacing and an exponentially amplified low--high ratio for throttling. After this signal problem is solved, the same budget-threshold framework converts the construction into \NOT, \NOR, and \NPURIFY gates and yields hardness throughout the entire non-vacuous approximation range. The appendices show that the same ideas survive when the different approximation margins are relaxed independently.

\paragraph{Related Work.}
Multiplicative pacing originates in bid-optimization dynamics for repeated ad auctions~\cite{BCI07}. First-price pacing is closely connected to Fisher markets and enjoys strong computational and structural properties~\cite{CKP22}; subsequent work studies computation and statistical inference in these settings~\cite{LK23,LK24}. In contrast, second-price pacing equilibria need not be unique, and optimizing welfare or revenue over them is NP-hard~\cite{CKSM22}. Chen, Kroer, and Kumar~\cite{CKK24} established $\PPAD$-hardness for inverse-polynomial approximation, while Chen and Li~\cite{CL25} proved the first hardness result for a fixed constant, namely $\gamma<1/3$, even when each bidder has positive value for at most four goods. Our endpoint construction trades this constant bidder-side support for support that depends on $1-\gamma$. Exact equilibria are nevertheless computable in polynomial time when either the number of bidders~\cite{YWL26} or the number of goods~\cite{HYWL26} is constant.

Chen, Kroer, and Kumar~\cite{CKK21T} introduced throttling equilibria for first- and second-price auctions and proved $\PPAD$-hardness for second-price throttling via approximate threshold games, even when each good has at most three positive bids. Our endpoint construction instead reduces directly from bounded-degree \PureCircuit and uses parameter-dependent support on signal goods. Deligkas, Fearnley, Hollender, and Melissourgos~\cite{DFHM24} introduced \PureCircuit and established its tight inapproximability properties; Chen and Li~\cite{CL25} gave the bounded-degree $\{\NOT,\NOR,\NPURIFY\}$ formulation used here.

Recent surveys cover the broader autobidding literature~\cite{ABB24,BBF24}. Related work studies auction design, strategic budgets, ROI constraints, and revenue guarantees~\cite{BDMMZ21,FLS24,LPSZ24,CIS26}. Chen, Morgenstern, and Yang~\cite{CMY26} obtain tractability and convergence under diffuse value distributions, complementing our worst-case setting.

\paragraph{Organization.}
\Cref{sec:preliminaries} introduces the two equilibrium notions and the bounded-degree \PureCircuit source problem. We then isolate the decoder and gate logic shared by both reductions in \Cref{sec:reduction-framework}. With this common layer in place, \Cref{sec:pacing-hardness,sec:throttling-hardness} prove the $\PPAD$-hardness results for pacing and throttling, respectively. Finally, \Cref{sec:discussion} summarizes the consequences and remaining questions. The independent-relaxation extensions are deferred to \Cref{app:independent-pacing,app:independent-throttling}.

\section{Preliminaries}
\label{sec:preliminaries}

\paragraph{Notation.}
For every positive integer $n$, write $[n]:=\{1,\ldots,n\}$.

\subsection{Second-Price Pacing Games}

The \emph{second-price pacing game} runs $m$ simultaneous single-good auctions for $n$ budget-constrained bidders. Bidder $i$ has value $v_{ij}\ge0$ for good $j$ and a budget $B_i>0$ across all goods. We assume that every bidder values at least one good and that every good has at least one positive valuation.

Bidder $i$ selects a multiplicative factor $\alpha_i \in [0,1]$ to scale all of her values uniformly, resulting in a bid of $\alpha_i v_{ij}$ on good $j$. Given a multiplier vector $\alpha \in [0,1]^n$, let $h_j(\alpha) := \max_i \alpha_i v_{ij}$ denote the highest scaled bid for good $j$, and let $p_j(\alpha)$ denote the second-highest scaled bid. A winner of good $j$ pays $p_j(\alpha)$ per unit. If the two largest bids tie, the highest and second-highest bids coincide; zero bids are also included when determining the second price. When there is only one bidder, we adopt the equivalent convention of appending a dummy zero bid, so $p_j(\alpha)=0$.

An allocation $x \in [0,1]^{n \times m}$ specifies fractions $x_{ij} \ge 0$ satisfying $\sum_i x_{ij} \le 1$ for each good $j$. Alternatively, $x_{ij}$ can be interpreted as the probability that bidder $i$ receives an indivisible good $j$. Under allocation $x$ and multiplier vector $\alpha$, bidder $i$'s total spending is $\spend_i(\alpha,x) := \sum_j x_{ij} p_j(\alpha)$.

A pacing equilibrium combines three ideas. Allocation must respect the scaled bids: only highest scaled bidders may receive a good, and every good with a positive highest bid must be fully allocated. Spending must remain within budget. Finally, pacing is complementary to budget slack: a bidder whose budget is not binding must use multiplier $1$. The formal definition records these requirements separately.

\begin{definition}[Pacing Equilibrium]
\label{def:pacing-equilibrium}
A pair $(\alpha,x)$ is a \emph{pacing equilibrium} if:
\begin{enumerate}[label=(\alph*)]
\item $x_{ij}>0$ implies $\alpha_i v_{ij}=h_j(\alpha)$;
\item $h_j(\alpha)>0$ implies $\sum_i x_{ij}=1$;
\item $\spend_i(\alpha,x)\le B_i$ for every bidder $i$; and
\item $\spend_i(\alpha,x)<B_i$ implies $\alpha_i=1$.
\end{enumerate}
\end{definition}

Pacing equilibria exist under the standard market assumptions~\cite{CKSM22}. Approximation affects only the complementarity condition: $\gamma$ determines how much budget slack is tolerated before a bidder must be fully unpaced, while allocation and budget feasibility remain exact. Chen and Li~\cite{CL25} define the parameter on $\gamma\in[0,1)$; we also record the natural value $\gamma=1$ so that the exact and vacuous endpoints align with the throttling model.

\begin{definition}[Approximate Pacing Equilibrium]
\label{def:approx-pacing}
Let $\gamma\in[0,1]$. A pair $(\alpha,x)$ is a \emph{$\gamma$-approximate pacing equilibrium} if:
\begin{enumerate}[label=(\alph*)]
\item $x_{ij}>0$ implies $\alpha_i v_{ij}=h_j(\alpha)$;
\item $h_j(\alpha)>0$ implies $\sum_i x_{ij}=1$;
\item $\spend_i(\alpha,x)\le B_i$ for every bidder $i$; and
\item $\spend_i(\alpha,x)<(1-\gamma)B_i$ implies $\alpha_i=1$.
\end{enumerate}
\end{definition}

At $\gamma=0$, this is exactly \Cref{def:pacing-equilibrium}. At $\gamma=1$, condition~(d) is vacuous, so the all-zero multiplier vector together with the zero allocation is immediately feasible.

\subsection{Second-Price Throttling Games}

The \emph{second-price throttling game} runs $m$ simultaneous single-good auctions for $n$ budget-constrained bidders with fixed nonnegative bids $(b_{ij})_{i,j}$ and positive budgets $(B_i)_i$. Throttling leaves bids unchanged and instead randomizes each bidder's admission to every auction. Bidder $i$ chooses a single participation probability $\theta_i\in[0,1]$ and is admitted independently of the other bidders within each auction. Correlation across different auctions does not affect expected spending. Among the admitted bidders, the highest bidder wins and pays the second-highest admitted bid, or zero if no other bidder participates. Each good has a fixed priority order for breaking bid ties.

The expected payment is obtained by conditioning on the runner-up. Let $i\succ_j\ell$ mean that $i$ ranks ahead of $\ell$ on good $j$, after applying the tie-breaking priority. For $i$ to win and pay $b_{\ell j}$, bidders $i$ and $\ell$ must participate and every bidder other than $i$ ranked above $\ell$ must be absent. Thus bidder $i$'s expected spending on good $j$ is
\begin{equation}
\label{eq:throttling-payment}
 q_{ij}(\theta)
 :=\sum_{\ell:\,i\succ_j\ell}
 b_{\ell j}\theta_i\theta_\ell
 \prod_{\substack{k\ne i:\,k\succ_j\ell}}(1-\theta_k).
\end{equation}
The quantity $S_i(\theta):=\sum_j q_{ij}(\theta)$ is bidder $i$'s expected spending across the market. The formula also covers bid ties because $\succ_j$ is a strict total order.

As in pacing, throttling combines budget feasibility with complementarity: expected spending cannot exceed the budget, and a bidder may be throttled only when the budget is exhausted.

\begin{definition}[Throttling Equilibrium]
\label{def:throttling-equilibrium}
A vector $\theta\in[0,1]^n$ is a \emph{throttling equilibrium} if, for every bidder $i$,
\begin{align*}
 S_i(\theta)&\le B_i,\\
 S_i(\theta)&<B_i\implies\theta_i=1.
\end{align*}
\end{definition}

The approximation parameter $\delta$ simultaneously relaxes two parts of complementarity. It lowers the spending threshold that activates the condition to a $(1-\delta)$ fraction of budget, and it lowers the required participation probability to $1-\delta$. Chen, Kroer, and Kumar~\cite{CKK21T} use these coupled conditions for positive error parameters. We write the same formula on $\delta\in[0,1]$, making $0$ the exact case and $1$ the vacuous endpoint, just as for pacing. \Cref{app:independent-throttling} later separates these two margins.

\begin{definition}[Approximate Throttling Equilibrium]
\label{def:approx-throttling}
For $\delta\in[0,1]$, a vector $\theta\in[0,1]^n$ is a \emph{$\delta$-approximate throttling equilibrium} if, for every bidder $i$,
\begin{align}
 S_i(\theta)&\le B_i,
 \label{eq:throttling-budget-feasibility}\\
 S_i(\theta)&<(1-\delta)B_i\implies\theta_i\ge1-\delta.
 \label{eq:throttling-complementarity}
\end{align}
\end{definition}

Thus a $0$-approximate throttling equilibrium is exactly a throttling equilibrium in the sense of \Cref{def:throttling-equilibrium}. At $\delta=1$, condition~\eqref{eq:throttling-complementarity} is vacuous, so the all-zero participation vector is immediately feasible.

To compare the sparsity of the two reductions on the same scale, we use the following support degrees.

\begin{definition}[Positive-Support Degrees]
\label{def:support-degrees}
Let $\mathcal I$ be either a pacing or a throttling instance, and define
\[
 w_{ij}:=
 \begin{cases}
  v_{ij},&\text{if $\mathcal I$ is a pacing instance},\\
  b_{ij},&\text{if $\mathcal I$ is a throttling instance}.
 \end{cases}
\]
The \emph{good-side positive-support degree} and \emph{bidder-side positive-support degree} of $\mathcal I$ are
\begin{align*}
 \goodsupport(\mathcal I)
 &:=\max_{j\in[m]}\left|\{i\in[n]:w_{ij}>0\}\right|,\\
 \biddersupport(\mathcal I)
 &:=\max_{i\in[n]}\left|\{j\in[m]:w_{ij}>0\}\right|.
\end{align*}
\end{definition}
In words, $\goodsupport(\mathcal I)$ is the largest number of bidders with a positive valuation or bid on any one good, while $\biddersupport(\mathcal I)$ is the largest number of goods on which any one bidder has a positive valuation or bid. When the instance is clear, we simply write $\goodsupport$ and $\biddersupport$.

\subsection{The \PureCircuit Problem}
\label{sec:purecircuit-problem}
Both reductions use the same bounded-degree \PureCircuit problem as their source. An instance consists of a node set $V$ and a collection of gates. A solution assigns each node a value in $\{0,1,\bot\}$, where $\bot$ denotes an unconstrained or indeterminate state, subject to the gate semantics below. Within each gate, all named inputs and outputs are distinct, so every interaction edge has distinct endpoints.
\begin{itemize}[leftmargin=2em]
\item $\NOT(u;v)$: Reverses a Boolean input and leaves $\bot$ unconstrained:
  \begin{equation*}
    \chi(u) \in \{0,1\} \implies \chi(v) = 1 - \chi(u).
  \end{equation*}
\item $\NOR(u,w;v)$: Outputs $1$ if both inputs are $0$, outputs $0$ if either input is $1$, and is unconstrained otherwise:
  \begin{align*}
    \chi(u) = 0 \land \chi(w) = 0 &\implies \chi(v) = 1, \\
    \chi(u) = 1 \lor \chi(w) = 1 &\implies \chi(v) = 0.
  \end{align*}
\item $\NPURIFY(u;v,w)$: Requires at least one output to be Boolean and forces both outputs to negate a Boolean input:
  \begin{align*}
    \{\chi(v), \chi(w)\} &\cap \{0,1\} \neq \emptyset, \\
    \chi(u) \in \{0,1\} &\implies \chi(v) = \chi(w) = 1 - \chi(u).
  \end{align*}
\end{itemize}

The \emph{interaction graph} has a directed edge from every gate input to every gate output; let $N^-(v)$ denote the in-neighbors of node $v$. We use the bounded-degree formulation of Chen and Li~\cite{CL25}, which is equivalent to the \PURIFY-based formulation of Deligkas, Fearnley, Hollender, and Melissourgos~\cite{DFHM24}. In this version, every node is the output of exactly one gate and has total degree at most three. Since every node therefore has in-degree at least one, its out-degree is at most two. The source theorem we use is the following.

\begin{theorem}[\PureCircuit Hardness~\cite{DFHM24,CL25}]\label{thm:pure-circuit}
  \PureCircuit restricted to $\{\NOT, \NOR, \NPURIFY\}$ gates is $\PPAD$-complete, even when every node is the output of exactly one gate and, more specifically, every node has $(d_{\mathrm{in}},d_{\mathrm{out}})\in\{(1,1),(2,1),(1,2)\}$.
\end{theorem}

\section{A Common Reduction Framework}
\label{sec:reduction-framework}
Although the pacing and throttling gadgets are different, they serve exactly the same logical purpose. Separating this shared circuit logic from the mechanism-specific signal construction makes both reductions shorter and easier to verify. This section proves the circuit layer once; the two hardness sections will then need only to establish the hypotheses of the common lemma.

Three mechanism-specific facts are sufficient. First, decoded $0$-states and $1$-states must generate well-separated signals. Second, the target budget must act as a threshold test: a sufficiently small aggregate signal forces the high state, whereas a sufficiently large aggregate signal forces the low state. Third, if the output of a unary gate remains indeterminate, then the predecessor signal must lie in a budget-dependent ambiguity interval. Once these properties are available, the semantics of \NOT, \NOR, and \NPURIFY follow from the same gate argument.

Fix thresholds $0\le \ell<h\le1$. For a node state $z\in[0,1]$, we define the \emph{threshold decoder}
\begin{equation}
\label{eq:abstract-threshold-decoder}
 \dec_{\ell,h}(z):=
 \begin{cases}
  0,&0\le z\le\ell,\\
  \bot,&\ell<z<h,\\
  1,&h\le z\le1.
 \end{cases}.
\end{equation}
For each circuit node $v$, the reduction creates a scalar state $z_v$ and decodes it as $\chi(v):=\dec_{\ell,h}(z_v)$. Each source node $u$ also produces one nonnegative \emph{abstract signal} $s_u$, shared by all edges leaving $u$. We do not require $s_u$ to equal a realized payment in every state; we only require it to control the target's budget as specified below. For a target $v$, let $A_v:=\sum_{u\in N^-(v)}s_u$ be the total incoming signal.

The first requirement is a \emph{signal gap}: decoded zeros must send a small signal, and decoded ones must send a large signal. Concretely, assume that there are constants $0\le a_0<a_1$ such that, for every node $u$,
\begin{equation}
\label{eq:abstract-signal-gap}
 \chi(u)=0\implies s_u\le a_0,
 \qquad
 \chi(u)=1\implies s_u\ge a_1.
\end{equation}
The second requirement is \emph{budget forcing}. For an output node $v$ with budget $B$, suppose the mechanism-specific analysis gives thresholds $\tau_1(B)$ and $\tau_0(B)$ such that
\begin{equation}
\label{eq:abstract-budget-test}
 A_v<\tau_1(B)\implies\chi(v)=1,
 \qquad
 A_v>\tau_0(B)\implies\chi(v)=0.
\end{equation}
The third requirement is only needed for purification. If a unary output $v$ with budget $B$ remains indeterminate, its predecessor signal must be trapped in an interval that depends on $B$:
\begin{equation}
\label{eq:abstract-ambiguity}
 \chi(v)=\bot\implies s_u\in J(B).
\end{equation}
The interval $J(B)$ is the range in which the budget test alone cannot determine the output. Once these three properties are established, the gate argument no longer depends on the auction model.

\begin{lemma}[Abstract Gate Simulation]
\label{lem:abstract-gate-simulation}
Under the preceding setup, suppose that \Cref{eq:abstract-signal-gap,eq:abstract-budget-test,eq:abstract-ambiguity} hold.
\begin{enumerate}[label=(\roman*)]
\item If a budget $B_G$ satisfies
\[
  2a_0<\tau_1(B_G)
  \qquad\text{and}\qquad
  \tau_0(B_G)<a_1,
\]
then assigning budget $B_G$ to the output makes the decoded states satisfy the constraints of both \NOT and \NOR gates.
\item Suppose that budgets $B_1,B_2$ each satisfy
\[
  a_0<\tau_1(B_i),
  \qquad
  \tau_0(B_i)<a_1,
\]
and $J(B_1)\cap J(B_2)=\varnothing$. Assigning $B_1$ and $B_2$ to the two outputs makes the decoded states satisfy the \NPURIFY constraint.
\end{enumerate}
\end{lemma}

\begin{proof}
For \NOT, we consider the two Boolean inputs. A zero input sends at most $a_0<\tau_1(B_G)$ and therefore forces the output to one; a one input sends at least $a_1>\tau_0(B_G)$ and forces the output to zero.

For \NOR, we apply the same thresholds to the sum of the two input signals. Two zeros contribute at most $2a_0<\tau_1(B_G)$ and force the output to one. If either input is one, that input alone contributes at least $a_1>\tau_0(B_G)$ and forces the output to zero. These are exactly the constrained cases of \NOR.

For \NPURIFY, a Boolean input similarly forces both outputs to its negation. If the input is $\bot$ and both outputs were also $\bot$, the shared source signal $s_u$ would have to lie in both $J(B_1)$ and $J(B_2)$, contradicting their disjointness. Hence at least one output is Boolean, as required.
\end{proof}

The lemma contains all gate-semantic reasoning needed in the remainder of the paper. Accordingly, each mechanism-specific proof follows the same sequence: construct the complete market, establish the signal gap, derive budget forcing and the ambiguity interval, verify the numerical gate inequalities, and finally decode an arbitrary approximate equilibrium.

\section{\texorpdfstring{$\PPAD$}{PPAD}-Hardness of Approximate Pacing}
\label{sec:pacing-hardness}

We begin with pacing. Because the common framework already handles the gate semantics, the pacing reduction has one mechanism-specific task: transmit a source multiplier to its targets without allowing the targets to induce comparable payment back to the source. We solve this task with a geometric wire, then show that every approximate pacing equilibrium satisfies the abstract gate conditions and therefore decodes to a valid \PureCircuit solution.

\begin{theorem}[Endpoint-Tight Pacing $\PPAD$-Hardness]
  \label{thm:pacing-hardness}
For every fixed $\gamma\in[0,1)$, computing a $\gamma$-approximate pacing equilibrium is $\PPAD$-hard, even for instances $\mathcal I$ satisfying
\[
 \goodsupport(\mathcal I)=2,
 \qquad
 \biddersupport(\mathcal I)=O((1-\gamma)^{-2}).
\]
\end{theorem}

Fix $\gamma\in[0,1)$ and write $\beta:=1-\gamma>0$. We reduce from the bounded-degree \PureCircuit instance of \Cref{thm:pure-circuit}. Set
\[
 L:=\frac{\beta^2}{40},
 \qquad
 U:=\frac12.
\]
For a multiplier profile $\alpha$, we decode node $v$ by
\begin{equation}
\label{eq:pacing-decoder}
 \chi(v):=\dec_{L,U}(\alpha_{b_v}).
\end{equation}
The choice $L=\Theta(\beta^2)$ is matched to the wire length below: with $K=\Theta(\beta^{-2})$ levels, the total backward leakage is of the same order as the low decoding region.

\subsection{Construction}
\label{subsec:pacing-construction}

We now construct the reduction from a \PureCircuit instance $\mathcal C$. Each circuit node $v$ becomes one \emph{variable bidder} $b_v$, and each interaction edge $u\to v$ becomes a geometric wire from $b_u$ to $b_v$. The wire has two simultaneous goals: the target should read the source multiplier through its second-price payment, while any reverse payment incurred by the source should be confined to a small geometric tail.

Set
\[
 T:=\frac L2,
 \qquad
 K:=\left\lceil\frac{30}{\beta^2}\right\rceil+1,
 \qquad
 \eta:=\frac3K,
\]
so $\eta<\beta^2/10$. At level $\ell\in[K]$, we define
\[
 t_\ell:=T4^{\ell-K},
\]
so $0<t_1<\cdots<t_K=T$.

\begin{construction}[Geometric Wire]
\label{constr:pacing-geometric-wire}
For every directed interaction edge $u\to v$, create goods $g^{uv}_1,\ldots,g^{uv}_K$. Only $b_u$ and $b_v$ value these goods, and for every $\ell\in[K]$ set
\begin{equation}
\label{eq:pacing-wire-valuations}
 v_{b_u,g^{uv}_\ell}=\frac1K,
 \qquad
 v_{b_v,g^{uv}_\ell}=\frac{1}{Kt_\ell}.
\end{equation}
\end{construction}

It remains to assign budgets to the variable bidders. We use
\[
 B_G:=\frac14
\]
for every \NOT and \NOR output. For an \NPURIFY gate, assign its two outputs the budgets
\begin{equation}
\label{eq:pacing-gate-budgets}
 B_1:=\frac\beta8,
 \qquad
 B_2:=\frac14,
\end{equation}
respectively. Every node is the output of exactly one gate, so this assigns one budget to every variable bidder. All unspecified valuations are zero. This completes the construction.

To see why the wire is one-way, fix an edge $u\to v$ and abbreviate $a:=\alpha_{b_u}$ and $b:=\alpha_{b_v}$. At level $\ell$, the two scaled bids are $a/K$ and $b/(Kt_\ell)$, and for $a>0$ they tie exactly when $b/a=t_\ell$. Because the thresholds $t_\ell$ grow geometrically, the levels on which the source can win form a suffix whose second prices decay geometrically.

\subsection{Proof of Correctness}
\label{subsec:pacing-correctness}

The correctness proof follows the common template. We first prove that the geometric wire transmits the source multiplier forward while leaking only $O(1/K)$ payment backward. We then translate these payment bounds into a budget threshold test and identify the ambiguity interval for an indeterminate unary output. The final step checks the numerical gate inequalities and invokes \Cref{lem:abstract-gate-simulation}.

\paragraph{Signal behavior.}
For one wire $u\to v$, let $W_{uv}$ denote its $K$ goods and abbreviate $a:=\alpha_{b_u}$ and $b:=\alpha_{b_v}$. The next lemma shows the two directional payment bounds that drive the reduction.

\begin{lemma}[One-Way Geometric Wire]
\label{lem:pacing-geometric-wire}
For every multiplier vector $\alpha$ and every allocation $x$ satisfying \Cref{def:approx-pacing}(a)--(b),
\begin{align}
  \sum_{j\in W_{uv}}x_{b_v j}p_j(\alpha)&\le a,
 \label{eq:pacing-target-payment-upper}\\
 \sum_{j\in W_{uv}}x_{b_u j}p_j(\alpha)&< \frac{4}{3K}.
 \label{eq:pacing-source-leakage}
\end{align}
Moreover, if $b>L$, then $b_v$ uniquely wins every good on the wire and
\begin{equation}
  \label{eq:pacing-exact-forward-payment}
  \sum_{j\in W_{uv}}x_{b_v j}p_j(\alpha)=a.
\end{equation}
\end{lemma}

\begin{proof}
If $b_v$ receives a positive fraction of a wire good, her payment per unit is at most the source bid $a/K$. Summing over the $K$ levels gives \Cref{eq:pacing-target-payment-upper}.

For the reverse direction, the claim is immediate when $a=0$. Assume $a>0$ and set $z:=b/a$. The source can receive level $\ell$ only if she is a highest bidder, which requires $t_\ell\ge z$. If no level satisfies this inequality, the source pays zero. Otherwise, let $k$ be the smallest index with $t_k\ge z$. Even if the source receives every good at levels $\ell\ge k$, her total payment is at most
\[
 \sum_{\ell=k}^K\frac{b}{Kt_\ell}
 =\frac{az}{Kt_k}\sum_{r=0}^{K-k}4^{-r}
 < \frac{4a}{3K}
 \le \frac{4}{3K},
\]
which proves \Cref{eq:pacing-source-leakage}.

Finally, suppose $b>L$. The conclusion is immediate if $a=0$, so assume $a>0$. Since $t_\ell\le T=L/2$ and $a\le1$,
\[
 \frac{b/(Kt_\ell)}{a/K}=\frac{b}{at_\ell}>\frac{L}{T}=2.
\]
Hence $b_v$ is the unique highest bidder on every level. Full allocation gives her every wire good, and she pays $a/K$ on each one. Summing over the $K$ levels proves \Cref{eq:pacing-exact-forward-payment}.
\end{proof}

Because the source circuit has bounded out-degree, the reverse leakage remains small even after all outgoing wires are combined.

\begin{corollary}
\label{cor:pacing-backward-leakage}
Every variable bidder pays less than $\eta$ on all wires corresponding to edges leaving her circuit node.
\end{corollary}

\begin{proof}
Every node has out-degree at most two by \Cref{thm:pure-circuit}. Applying \Cref{eq:pacing-source-leakage} to each outgoing wire gives total payment less than
\[
 \frac{8}{3K}<\frac3K=\eta.
\]
\end{proof}

\paragraph{Budget forcing and ambiguity.}
For a node $v$, we define the aggregate incoming signal
\[
 A_v:=\sum_{u\in N^-(v)}\alpha_{b_u}.
\]
The next lemma gives the desired high--low response. A small $A_v$ leaves enough budget slack to force multiplier $1$, while a large $A_v$ would already violate the budget if the output multiplier remained above the low threshold.

\begin{lemma}[Pacing Budget Forcing]
\label{lem:pacing-budget-forcing}
If bidder $b_v$ has budget $B_v$, every $\gamma$-approximate pacing equilibrium satisfies
\begin{align*}
 A_v+\eta\le\beta B_v &\implies \alpha_{b_v}=1,\\
 A_v>B_v &\implies \alpha_{b_v}\le L.
\end{align*}
\end{lemma}

\begin{proof}
By \Cref{eq:pacing-target-payment-upper}, bidder $b_v$ pays at most $\alpha_{b_u}$ on each incoming wire $u\to v$. By \Cref{cor:pacing-backward-leakage}, she pays less than $\eta$ across all outgoing wires. These are all goods that she values, so
\[
 \spend_{b_v}(\alpha,x)<A_v+\eta\le\beta B_v.
\]
Condition~\Cref{def:approx-pacing}(d) therefore forces $\alpha_{b_v}=1$.

For the second implication, suppose $A_v>B_v$ but $\alpha_{b_v}>L$. Then \Cref{eq:pacing-exact-forward-payment} applies to every incoming wire, so $b_v$ pays exactly $A_v>B_v$ on those wires alone, contradicting budget feasibility.
\end{proof}

The same threshold test also pins down the only predecessor range compatible with an indeterminate unary output. We will use the following numerical bounds repeatedly; in particular, they imply $L<U$:
\begin{equation}
  \label{eq:pacing-numerical-bounds}
  L+\eta<\beta^2/8,
  \qquad
  2L+\eta<3\beta^2/20.
\end{equation}

\begin{lemma}[Pacing Ambiguity Interval]
\label{lem:pacing-ambiguity-interval}
Consider a unary-gate output $v$ with budget $B$ and predecessor $u$. If $\chi(v)=\bot$, then
\[
 \alpha_{b_u}\in J(B):=(\beta B-\eta,B].
\]
\end{lemma}

\begin{proof}
Write $a:=\alpha_{b_u}$. If $a+\eta\le\beta B$, then \Cref{lem:pacing-budget-forcing} forces $\alpha_{b_v}=1$, contradicting $\chi(v)=\bot$. Hence $a>\beta B-\eta$. If $a>B$, the same lemma forces $\alpha_{b_v}\le L$, again contradicting $\chi(v)=\bot$. Therefore $a\in(\beta B-\eta,B]$.
\end{proof}

\paragraph{Gate verification and decoding.}
We now match the pacing construction to the common framework. Set
\[
 z_u:=\alpha_{b_u},
 \qquad
 s_u:=\alpha_{b_u},
 \qquad
 a_0:=L,
 \qquad
 a_1:=U=\frac12.
\]
By \Cref{lem:pacing-budget-forcing}, the abstract thresholds are
\[
 \tau_1(B)=\beta B-\eta,
 \qquad
 \tau_0(B)=B,
\]
and \Cref{lem:pacing-ambiguity-interval} gives $J(B)=(\beta B-\eta,B]$.

For the common \NOT/\NOR budget $B_G=1/4$, \Cref{eq:pacing-numerical-bounds} yields
\[
 2L+\eta<\frac{3\beta^2}{20}<\frac\beta4=\beta B_G,
\]
so $2a_0<\tau_1(B_G)$, while $\tau_0(B_G)=1/4<a_1$. For the two \NPURIFY budgets in \Cref{eq:pacing-gate-budgets}, the same bounds give $L+\eta<\beta B_i$, hence $a_0<\tau_1(B_i)$, and $\tau_0(B_i)=B_i<a_1$ for $i\in\{1,2\}$. Finally,
\[
 J(B_1)=\left(\frac{\beta^2}{8}-\eta,\frac\beta8\right],
 \qquad
 J(B_2)=\left(\frac\beta4-\eta,\frac14\right]
\]
are disjoint because $\eta<\beta^2/10\le\beta/10$ implies $\beta/8<\beta/4-\eta$. Thus every hypothesis of \Cref{lem:abstract-gate-simulation} is satisfied.

We can now decode an arbitrary approximate equilibrium and complete the reduction.

\begin{proof}[Proof of~\Cref{thm:pacing-hardness}]
Take the market constructed from a bounded-degree \PureCircuit instance $\mathcal C$, and let $(\alpha,x)$ be any $\gamma$-approximate pacing equilibrium. Decode every circuit node using \Cref{eq:pacing-decoder}. If $\chi(u)=0$, then the source signal $s_u=\alpha_{b_u}$ is at most $L=a_0$; if $\chi(u)=1$, then $s_u\ge U=a_1$. \Cref{lem:pacing-budget-forcing,lem:pacing-ambiguity-interval} provide the corresponding abstract thresholds and intervals, and the parameter verification above shows that the budgets assigned in the construction satisfy all conditions of \Cref{lem:abstract-gate-simulation}. Applying that lemma to every gate, we obtain a decoded assignment $\chi$ satisfying every \NOT, \NOR, and \NPURIFY constraint of $\mathcal C$. Thus a valid solution to the \PureCircuit instance can be extracted from every approximate pacing equilibrium in polynomial time.

Every wire good has exactly two positive valuations, hence $\goodsupport=2$. Since each circuit node has total interaction degree at most three, each variable bidder is incident to at most three wires and therefore values at most $3K$ goods. Consequently,
\[
 \biddersupport\le3K=O((1-\gamma)^{-2}).
\]
The interaction graph has $O(|V|)$ edges, so the reduction creates $|V|$ bidders and $O(|V|K)$ goods. Together with \Cref{thm:pure-circuit}, this proves the theorem.
\end{proof}

Because the construction exposes its dependence on $1-\gamma$ explicitly, the same reduction also remains polynomial when the distance to the endpoint is inverse-polynomial rather than constant.

\begin{corollary}[Inverse-Polynomial Pacing Gap]
  \label{cor:pacing-inverse-polynomial-gap}
  For every fixed nonnegative integer $c$, the reduction remains polynomial when, on an $N$-node source instance, $1-\gamma=N^{-c}$. Indeed, $K=O(N^{2c})$, and the geometric wire values have $O(K+\log N)$ bits. Consequently, in the uniform search problem in which the parameter $\gamma\in[0,1)$ is part of the input, $\PPAD$-hardness persists even when $1-\gamma$ is inverse-polynomial in the total binary input length.
\end{corollary}

\section{\texorpdfstring{$\PPAD$}{PPAD}-Hardness of Approximate Throttling}
\label{sec:throttling-hardness}

We next turn to throttling. The gate logic remains unchanged, but the signal obstacle is different. Variable bidders create no backward payment leakage, yet every incoming payment is attenuated by the target's own participation probability. A direct analogue of the pacing wire therefore cannot create a robust low--high gap. We instead build a blocker amplifier whose joint absence probability magnifies the source state, and then verify that the amplified signal satisfies the same abstract gate conditions.

\begin{theorem}[Endpoint-Tight Throttling $\PPAD$-Hardness]
\label{thm:throttling-hardness}
For every fixed $\delta\in(0,1)$, computing a $\delta$-approximate throttling equilibrium is $\PPAD$-hard. Writing $H:=1-\delta$, hardness holds even for instances $\mathcal I$ satisfying
\[
 \goodsupport(\mathcal I),\ \biddersupport(\mathcal I)
 =O\left(H^{-1}(1+\log(1/H))\right).
\]
\end{theorem}

Fix $\delta\in(0,1)$ and set $H:=1-\delta\in(0,1)$. We reduce from the same bounded-degree \PureCircuit problem. Set
\[
 L:=\frac{H^3}{8},
\]
and decode a participation vector $\theta$ by
\begin{equation}
\label{eq:throttling-decoder}
 \chi(v):=\dec_{L,H}(\theta_v).
\end{equation}
The cubic low scale leaves enough budget slack to force the blocker response used below.

\subsection{Construction}
\label{subsec:throttling-construction}

For every circuit node $v$, create one \emph{variable bidder}, also denoted $v$. A source node $u$ with positive out-degree receives one bundle of blockers, and the same bundle is reused on every edge leaving $u$. This sharing is important for \NPURIFY: both outputs must observe exactly the same predecessor signal, because the purification argument excludes simultaneous indeterminacy by placing that common signal in two disjoint ambiguity intervals.

Set
\[
 \lambda:=\frac H4,
 \qquad
 p:=\frac L{12},
 \qquad
 C:=\frac{H^2}{4}.
\]
Choose the smallest positive integer $K$ satisfying
\begin{equation}
\label{eq:throttling-amplifier-size}
 4\left(\frac{1-H}{1-H/4}\right)^K<H^2L.
\end{equation}
Writing $r:=(1-H)/(1-H/4)\in[0,1)$, we have $K=1$ when $H=1$. For $H\in(0,1)$,
\[
 r=1-\frac{3H/4}{1-H/4}\le1-\frac{3H}{4}\le e^{-3H/4},
\]
so it is enough to take
\[
 K>\frac{4}{3H}\left(\log 32+5\log\frac1H\right).
\]
Hence the least feasible value satisfies
\begin{equation}
\label{eq:throttling-K-bound}
 K=O\left(H^{-1}(1+\log(1/H))\right).
\end{equation}
For fixed $\delta$, this is a constant-size amplifier; the explicit dependence is needed only to track the approach to the endpoint $\delta=1$.

For every source node $u$ of positive out-degree, we create blocker bidders $z_{u,1},\ldots,z_{u,K}$, each with budget $C$.

\begin{construction}[Blocker Control Good]
\label{constr:throttling-control-good}
For every such $u$ and every $q\in[K]$, we create a good $c_{u,q}$. Blocker $z_{u,q}$ bids $2$, variable bidder $u$ bids $1$, and all other bids are zero.
\end{construction}

For every interaction edge $u\to v$, create one signal good $s_{uv}$ whose bid ordering is blockers above target above source.

\begin{construction}[Signal Good]
\label{constr:throttling-signal-good}
On good $s_{uv}$, every blocker $z_{u,q}$ bids $3p$, the target $v$ bids $2p$, and the source $u$ bids $p$. All other bids are zero.
\end{construction}

The two types of goods play complementary roles. Control goods make each blocker react to the source state; the signal good then exposes the source to the target only when every blocker is absent. To calibrate the gate budgets against the resulting low--high separation, define the signal scales
\begin{equation}
\label{eq:throttling-signal-scales}
 a_0:=pL(1-H)^K,
 \qquad
 a_1:=pH(1-H/4)^K,
\end{equation}
and set
\begin{equation}
\label{eq:throttling-gate-budgets}
 B_1:=\frac{L^2a_1}{2},
 \qquad
 B_2:=\frac{3La_1}{4}.
\end{equation}
We assign budget $B_1$ to every \NOT and \NOR output. For an \NPURIFY gate, we assign its two outputs budgets $B_1$ and $B_2$, respectively. Every blocker keeps budget $C$, and all unspecified bids are zero. This completes the throttling construction.

\subsection{Proof of Correctness}
\label{subsec:throttling-correctness}

The proof again follows the common template. First, the control goods force opposite blocker responses in the low and high source states, and the product of blocker-absence probabilities turns this response into a separated signal. Second, an exact identity for variable-bidder spending converts the signal gap into budget forcing and an ambiguity interval. Finally, we verify the gate inequalities and invoke the common gate lemma to decode the source circuit.

\paragraph{Signal behavior.}
The control good is designed to make a blocker move in the opposite direction from its source. When the source is low, the blocker remains far enough below budget that complementarity forces high participation. When the source is high, any blocker participation above $\lambda$ would already violate budget feasibility.

\begin{lemma}[Blocker Response]
\label{lem:throttling-blocker-response}
In every $\delta$-approximate throttling equilibrium, for every node $u$ of positive out-degree and every $q\in[K]$,
\begin{align*}
 \theta_u\le L&\implies\theta_{z_{u,q}}\ge H,\\
 \theta_u\ge H&\implies\theta_{z_{u,q}}\le\lambda.
\end{align*}
\end{lemma}

\begin{proof}
Fix one blocker and write $z:=z_{u,q}$. On the control good $c_{u,q}$, bidder $z$ is the highest bidder and pays $1$ exactly when both $z$ and $u$ participate. Her expected control payment is therefore $\theta_z\theta_u$.

By \Cref{thm:pure-circuit}, node $u$ has out-degree at most two, so $z$ bids on at most two signal goods. On each such good, her payment is at most $3p$ whenever she participates. Hence
\[
 \theta_z\theta_u\le S_z(\theta)\le\theta_z(\theta_u+6p).
\]
If $\theta_u\le L$, then
\[
 S_z(\theta)\le L+6p=\frac{3L}{2}=\frac{3H^3}{16}<\frac{H^3}{4}=HC.
\]
Condition~\eqref{eq:throttling-complementarity} therefore forces $\theta_z\ge H$.

If instead $\theta_u\ge H$ and $\theta_z>\lambda$, then the control payment alone satisfies
\[
 \theta_z\theta_u>\lambda H=\frac{H^2}{4}=C,
\]
contradicting budget feasibility. Thus $\theta_z\le\lambda$.
\end{proof}

For a source $u$ with positive out-degree, we define the payment coefficient shared by all edges leaving $u$ as
\begin{equation}
\label{eq:throttling-signal-coefficient}
 y_u:=p\theta_u\prod_{q=1}^K(1-\theta_{z_{u,q}}).
\end{equation}
Conditional on a target $v$ participating, $y_u$ is exactly her expected payment on $s_{uv}$: the target wins only when all blockers are absent, and conditional on that event she pays $p$ precisely when the source participates. Therefore
\begin{equation}
\label{eq:throttling-exact-signal-payment}
 q_{v,s_{uv}}(\theta)=\theta_v y_u.
\end{equation}
The source never pays on a signal good, because whenever she wins all higher bidders are absent and the second price is zero; she also never pays on a control good.

The signal coefficient multiplies the absence probabilities of all blockers, so these individual response bounds compound across the bundle. The definitions of $a_0$ and $a_1$ in \Cref{eq:throttling-signal-scales} are exactly the resulting worst-case low and high signal levels.

\begin{lemma}[Amplified Signal Gap]
\label{lem:throttling-signal-gap}
Every $\delta$-approximate throttling equilibrium satisfies
\begin{align}
 \theta_u\le L&\implies y_u\le a_0,
 \label{eq:throttling-low-signal}\\
 \theta_u\ge H&\implies y_u\ge a_1.
 \label{eq:throttling-high-signal}
\end{align}
Moreover, the choice of $K$ implies
\begin{equation}
\label{eq:throttling-master-inequality}
 4a_0<HL^2a_1.
\end{equation}
\end{lemma}

\begin{proof}
If $\theta_u\le L$, then \Cref{lem:throttling-blocker-response} gives $\theta_{z_{u,q}}\ge H$ for every $q$. Substituting these bounds into \Cref{eq:throttling-signal-coefficient} gives $y_u\le pL(1-H)^K=a_0$. If $\theta_u\ge H$, every blocker probability is at most $H/4$, and therefore $y_u\ge pH(1-H/4)^K=a_1$.

Finally, multiplying \Cref{eq:throttling-amplifier-size} by $pL(1-H/4)^K$ yields $4pL(1-H)^K<pH^2L^2(1-H/4)^K$.
The two sides are $4a_0$ and $HL^2a_1$, respectively, proving \Cref{eq:throttling-master-inequality}. In particular $a_0<a_1$, because $HL^2<1$.
\end{proof}

\paragraph{Budget forcing and ambiguity.}
For node $v$, let
\[
 A_v:=\sum_{u\in N^-(v)}y_u.
\]
Because a variable bidder pays only on incoming signal goods, her spending is exactly her participation probability times this aggregate signal.

\begin{corollary}[Exact Variable Spending]
\label{cor:throttling-exact-variable-spending}
Every variable bidder $v$ satisfies
\begin{equation}
\label{eq:throttling-variable-spending}
 S_v(\theta)=\theta_v A_v.
\end{equation}
\end{corollary}

\begin{proof}
For every incoming edge $u\to v$, \Cref{eq:throttling-exact-signal-payment} gives payment $\theta_v y_u$. As observed above, a variable bidder pays zero on all outgoing signal and control goods. Summing the incoming payments proves the identity.
\end{proof}

This exact identity is the throttling analogue of the pacing budget test. A small aggregate signal leaves enough slack to activate complementarity and force high participation, whereas a large aggregate signal makes any participation above $L$ violate the budget.

\begin{lemma}[Throttling Budget Forcing]
\label{lem:throttling-budget-forcing}
For any variable bidder $v$ with budget $B$, every $\delta$-approximate throttling equilibrium satisfies
\begin{align*}
 A_v<HB &\implies \theta_v\ge H,\\
 A_v>B/L &\implies \theta_v\le L.
\end{align*}
\end{lemma}

\begin{proof}
Since $\theta_v\le1$, \Cref{eq:throttling-variable-spending} gives $S_v(\theta)\le A_v$. Thus $A_v<HB$ satisfies the premise of \eqref{eq:throttling-complementarity}, forcing $\theta_v\ge H$.

Conversely, if $A_v>B/L$ and $\theta_v>L$, then
\[
 S_v(\theta)=\theta_vA_v>L A_v>B,
\]
contradicting budget feasibility. Hence $\theta_v\le L$.
\end{proof}

For unary outputs, the same exact spending identity determines the full range of predecessor signals compatible with an indeterminate output.

\begin{lemma}[Throttling Ambiguity Interval]
\label{lem:throttling-ambiguity-interval}
Consider a unary-gate output $v$ with budget $B$ and predecessor $u$. If $\chi(v)=\bot$, then
\[
 y_u\in J(B):=(B,B/L).
\]
\end{lemma}

\begin{proof}
Write $\theta:=\theta_v$ and $y:=y_u$. By \Cref{eq:throttling-decoder}, the condition $\chi(v)=\bot$ means $L<\theta<H$. Since $\theta<H$, the contrapositive of \eqref{eq:throttling-complementarity} gives $\theta y\ge HB$, and therefore $y\ge HB/\theta>B$.
Budget feasibility gives $\theta y\le B$; because $\theta>L$, we also have $y\le B/\theta< B/L$.
Thus $y\in(B,B/L)$.
\end{proof}

\paragraph{Gate verification and decoding.}
We instantiate the abstract framework with
\[
 z_u:=\theta_u,
 \qquad
 s_u:=y_u.
\]
By \Cref{lem:throttling-signal-gap}, decoded zeros and ones send signals bounded by $a_0$ and $a_1$. By \Cref{lem:throttling-budget-forcing},
\[
 \tau_1(B)=HB,
 \qquad
 \tau_0(B)=B/L,
\]
while \Cref{lem:throttling-ambiguity-interval} gives $J(B)=(B,B/L)$.

For the common budget $B_1$, the master inequality gives
\[
 2a_0<HB_1,
 \qquad
 \frac{B_1}{L}=\frac{La_1}{2}<a_1.
\]
Thus $B_1$ satisfies the \NOT/\NOR conditions of \Cref{lem:abstract-gate-simulation}. For \NPURIFY, we also have $a_0<HB_1<HB_2$ and
\[
 \frac{B_1}{L}=\frac{La_1}{2}<a_1,
 \qquad
 \frac{B_2}{L}=\frac{3a_1}{4}<a_1.
\]
Finally,
\[
 \frac{B_1}{L}=\frac{La_1}{2}<\frac{3La_1}{4}=B_2,
\]
so $J(B_1)\cap J(B_2)=\varnothing$. Hence the two gate budgets assigned in the construction satisfy every hypothesis of the abstract gate lemma.

We can now decode the market equilibrium and finish the reduction.

\begin{proof}[Proof of \Cref{thm:throttling-hardness}]
Take the market constructed from a bounded-degree \PureCircuit instance $\mathcal C$, and let $\theta$ be any $\delta$-approximate throttling equilibrium. We decode every circuit node using \Cref{eq:throttling-decoder}. By \Cref{lem:throttling-signal-gap}, a decoded zero sends signal at most $a_0$, while a decoded one sends signal at least $a_1$. \Cref{lem:throttling-budget-forcing,lem:throttling-ambiguity-interval} give the abstract thresholds $\tau_1(B)=HB$, $\tau_0(B)=B/L$, and $J(B)=(B,B/L)$. The parameter verification above therefore allows us to apply \Cref{lem:abstract-gate-simulation} to every gate. For \NPURIFY, the two outputs see the same coefficient $y_u$ because the construction reuses the source's blocker bundle. We conclude that the decoded assignment $\chi$ satisfies the entire source circuit, so every approximate throttling equilibrium yields a valid solution of $\mathcal C$ in polynomial time.

A control good has two positive bids and a signal good has $K+2$, hence $\goodsupport\le K+2$. A variable bidder bids on the $K$ control goods of her blocker bundle, when she has one, and on at most three incident signal goods; a blocker bids on one control good and at most two signal goods. Therefore $\biddersupport\le K+3$. Since the interaction graph has $O(|V|)$ edges, the market contains $O(|V|K)$ bidders and goods. Combining these facts with \Cref{eq:throttling-K-bound} gives
\[
 \goodsupport,\ \biddersupport
 =O\left(H^{-1}(1+\log(1/H))\right).
\]
Together with \Cref{thm:pure-circuit}, this proves the theorem.
\end{proof}

As in the pacing case, the explicit dependence on $H=1-\delta$ shows that the reduction remains polynomial even when the endpoint gap is inverse-polynomial.

\begin{corollary}[Inverse-Polynomial Throttling Gap]
\label{cor:throttling-inverse-polynomial-gap}
For every fixed positive integer $c$, the reduction remains polynomial when, on an $N$-node source instance with $N\ge2$, $1-\delta=N^{-c}$. Here $K=O(N^c\log N)$, and the rational powers appearing in the signal scales and budgets have polynomial bit length. Consequently, in the uniform search problem in which the parameter $\delta\in(0,1)$ is part of the input, $\PPAD$-hardness persists even when $1-\delta$ is inverse-polynomial in the total binary input length.
\end{corollary}

\section{Discussion}
\label{sec:discussion}

The main conclusion is the same for pacing and throttling: approximation does not create a nontrivial tractable regime before the complementarity condition becomes vacuous. Pacing is $\PPAD$-hard for every fixed $\gamma<1$, and throttling is $\PPAD$-hard for every fixed $0<\delta<1$, while at parameter $1$ the all-zero solution is immediate.
The common reduction framework isolates these mechanism-specific effects from the gate logic, making the endpoint phenomenon visible at the level of signal transmission rather than at the level of individual circuit gadgets. We believe it may be of independent interest.

Several questions remain. First, exact second-price throttling requires separate treatment because exact equilibria can be irrational~\cite{CKK21T}. The construction above remains algebraically valid at $\delta=0$, but we do not formulate this endpoint as a standard finite-output $\PPAD$ search problem. A natural question is to determine the exact real-algebraic complexity of the problem, for example whether it is $\mathsf{FIXP}$-complete under a standard algebraic representation.

Second, it is unclear how sparse endpoint-tight hard instances can be. For throttling, the case in which every good has at most two positive bids admits a polynomial-time approximation algorithm for every fixed positive error~\cite{CKK21T}, while the known hardness construction uses at most three positive bids per good. Our signal goods instead use $K+2$ positive bidders. More broadly, one can ask whether endpoint-tight hardness is possible with constant support independent of the distance to the endpoint, with uniformly bounded numerical ratios, or even over a fixed finite alphabet of bids, values, and budgets.

Third, we also wonder whether the uniform versions of our reductions could achieve inverse-exponential endpoint gaps, instead of the current polynomially small ones.

\paragraph{AI Disclosure.}
The author used GPT-5.6 Sol (with Extra High model) to assist with the construction and analysis of the reductions, and Gemini 3.6 Flash to assist with presentation. The correctness of all proofs was thoroughly verified by the author.

\printbibliography

\clearpage

\appendix
\section{Independent Relaxations for Pacing}
\label{app:independent-pacing}

The main pacing theorem keeps winner selection exact and relaxes only the complementarity trigger. This appendix separates three approximation margins: winner selection, the spending threshold that activates complementarity, and the multiplier required once complementarity is activated. Hardness persists for every fixed triple below $1$. The only genuinely new issue is approximate winner eligibility: even when the target has the highest scaled bid, the source may remain eligible for allocation and therefore incur an additional payment. Widening the spacing between consecutive wire levels ensures that this extra leakage can occur at at most one level.

\begin{definition}[Independent Pacing Approximation~\cite{CL25}]
\label{def:independent-pacing}
Let $\sigma,\gamma,\tau\in[0,1)$. A pair $(\alpha,x)$ is a \emph{$(\sigma,\gamma,\tau)$-approximate pacing equilibrium} if:
\begin{enumerate}[label=(\alph*)]
\item $x_{ij}>0$ implies $\alpha_i v_{ij}\ge(1-\sigma)h_j(\alpha)$;
\item $h_j(\alpha)>0$ implies $\sum_i x_{ij}=1$;
\item $\spend_i(\alpha,x)\le B_i$ for every bidder $i$; and
\item $\spend_i(\alpha,x)<(1-\gamma)B_i$ implies $\alpha_i\ge1-\tau$.
\end{enumerate}
\end{definition}

The three parameters have distinct roles: $\sigma$ relaxes winner selection, $\gamma$ relaxes the spending trigger, and $\tau$ relaxes the multiplier required after the trigger. Prices remain the profile-wide second-highest scaled bids from the preliminaries, even when the allocated bidder is only approximately highest; in particular, an allocated bidder who is second-highest may pay her own scaled bid. We call bidder $i$ \emph{approximately eligible} for good $j$ when $\alpha_i v_{ij}\ge(1-\sigma)h_j(\alpha)$. This enlarged eligibility set is precisely the source of the extra leakage handled below.

Setting $\sigma=\tau=0$ recovers \Cref{def:approx-pacing}. The two-parameter approximation of Chen, Kroer, and Kumar is obtained by coupling the last two relaxations: their parameters $(\delta,\rho)$ correspond to $(\sigma,\gamma,\tau)=(\delta,\rho,\rho)$~\cite{CKK24}.

\begin{theorem}[Independent-Relaxation Pacing $\PPAD$-Hardness]
\label{thm:independent-pacing}
For every fixed $\sigma,\gamma,\tau\in[0,1)$, computing a pair $(\alpha,x)$ satisfying \Cref{def:independent-pacing} is $\PPAD$-hard, even when
\[
 \goodsupport(\mathcal I)=2,
 \qquad
 \biddersupport(\mathcal I)
 =O\left(((1-\gamma)^2(1-\tau))^{-1}\right).
\]
\end{theorem}

Fix $\sigma,\gamma,\tau\in[0,1)$ and write
\[
 s:=1-\sigma,
 \qquad
 \beta:=1-\gamma,
 \qquad
 h:=1-\tau,
\]
so $s,\beta,h\in(0,1]$. We reduce from the same bounded-degree \PureCircuit problem as in the main text.

\subsection{Construction}
\label{subsec:independent-pacing-construction}

The construction keeps the topology of the main pacing reduction and changes only the decoding scale, the wire spacing, and the corresponding budgets. Set
\begin{equation}
\label{eq:independent-pacing-decoding-thresholds}
 U:=\frac h2,
 \qquad
 L:=\frac{\beta^2h}{512},
 \qquad
 T:=\frac{sL}{2},
\end{equation}
and decode node $v$ by
\begin{equation}
\label{eq:independent-pacing-decoder}
 \chi(v):=\dec_{L,U}(\alpha_{b_v}).
\end{equation}
For the wire, set
\begin{equation}
\label{eq:independent-pacing-wire-parameters}
 Q:=\frac4s,
 \qquad
 K:=\left\lceil\frac{8192}{\beta^2h}\right\rceil+1,
 \qquad
 \eta:=\frac6K.
\end{equation}
Then
\begin{equation}
\label{eq:independent-pacing-leakage-bound}
 \eta<\frac{3\beta^2h}{4096}.
\end{equation}

For every circuit node $v$, create one variable bidder $b_v$. For every interaction edge $u\to v$, create goods $g^{uv}_1,\ldots,g^{uv}_K$ and write $W_{uv}:=\{g^{uv}_1,\ldots,g^{uv}_K\}$ for the corresponding wire. At level $\ell\in[K]$, let
\begin{equation}
\label{eq:independent-pacing-wire-thresholds}
 t_\ell:=TQ^{\ell-K},
\end{equation}
and give only $b_u$ and $b_v$ positive values:
\begin{equation}
\label{eq:independent-pacing-wire-valuations}
 v_{b_u,g^{uv}_\ell}=\frac1K,
 \qquad
 v_{b_v,g^{uv}_\ell}=\frac{1}{Kt_\ell}.
\end{equation}
Because $Q=4/s>1/s$, consecutive thresholds are far enough apart that the approximate-eligibility window used in the proof contains at most one level.

Finally, assign gate budgets
\begin{equation}
\label{eq:independent-pacing-gate-budgets}
 B_G:=\frac h4,
 \qquad
 B_1:=\frac{\beta h}{32},
 \qquad
 B_2:=\frac h4.
\end{equation}
Every \NOT and \NOR output receives budget $B_G$, while the two outputs of an \NPURIFY gate receive $B_1$ and $B_2$, respectively. Every node is the output of exactly one gate, so this assigns one budget to each variable bidder. All unspecified valuations are zero. This completes the market construction.

\subsection{Proof of Correctness}
\label{subsec:independent-pacing-correctness}

The proof mirrors the main pacing argument. Wider spacing first controls the extra leakage created by approximate winner eligibility. The resulting wire bounds then yield budget forcing and an ambiguity interval, after which the same abstract gate lemma completes the decoding argument.

\paragraph{Signal behavior.}
Approximate winner eligibility affects only the reverse-payment bound. The target can be strictly highest while the source is still eligible to receive the good, but the choice $Q>1/s$ makes the corresponding eligibility window too narrow to contain more than one wire level.

\begin{lemma}[Robust Geometric Wire]
\label{lem:independent-pacing-geometric-wire}
For every multiplier vector $\alpha$ and every allocation $x$ satisfying \Cref{def:independent-pacing}(a)--(b), the wire $u\to v$ satisfies
\begin{align}
 \sum_{j\in W_{uv}}x_{b_v j}p_j(\alpha)&\le a,
 \label{eq:independent-pacing-target-payment-upper}\\
 \sum_{j\in W_{uv}}x_{b_u j}p_j(\alpha)&<\frac3K,
 \label{eq:independent-pacing-source-leakage}
\end{align}
where $a:=\alpha_{b_u}$. Moreover, if $b:=\alpha_{b_v}>L$, then $b_v$ is the unique approximately eligible bidder on every wire good, receives every good in full, and pays exactly $a$ on the wire.
\end{lemma}

\begin{proof}
Because a wire good has exactly two positive valuations, its second price is the smaller of the two scaled bids. The target therefore pays at most the source bid $a/K$ on each level, proving \Cref{eq:independent-pacing-target-payment-upper}.

For the reverse direction, the claim is immediate when $a=0$, so assume $a>0$ and set $z:=b/a$. First consider levels on which the source is a highest bidder. These satisfy $t_\ell\ge z$. If there is no such level, their contribution is zero. Otherwise, let $k$ be the smallest index with $t_k\ge z$. Even if the source receives every level $\ell\ge k$, her total payment from these levels is at most
\[
 \sum_{\ell=k}^K\frac{b}{Kt_\ell}
 \le\frac aK\sum_{r=0}^{\infty}Q^{-r}
 =\frac aK\frac{Q}{Q-1}
 \le\frac{4a}{3K}.
\]

It remains to consider levels on which the target is strictly highest but the source is still approximately eligible. These two conditions are equivalent to
\begin{equation}
\label{eq:independent-pacing-eligibility-window}
 sz\le t_\ell<z.
\end{equation}
There is at most one such level: two thresholds in this interval would have ratio less than $1/s$, whereas consecutive thresholds have ratio $Q=4/s$. On this exceptional level, the second price is the source bid $a/K$. Hence the source's total payment is less than
\[
 \frac{4a}{3K}+\frac aK<\frac3K,
\]
which proves \Cref{eq:independent-pacing-source-leakage}.

Finally, suppose $b>L$. Since $t_\ell\le T=sL/2$ and $a\le1$,
\[
 \frac{b}{Kt_\ell}>\frac{2}{sK}\ge\frac{2a}{sK}.
\]
Thus the target is strictly highest and the source bid is less than an $s$ fraction of the target bid. The target is therefore the unique approximately eligible bidder, receives each wire good in full, and pays $a/K$ at every level. Summing over all $K$ levels gives total payment $a$.
\end{proof}

\begin{corollary}[Robust Backward Leakage]
\label{cor:independent-pacing-backward-leakage}
Every variable bidder pays less than $\eta$ across all wires leaving her circuit node.
\end{corollary}

\begin{proof}
Every source node has out-degree at most two. By \Cref{eq:independent-pacing-source-leakage}, the total payment on its outgoing wires is therefore less than $6/K=\eta$.
\end{proof}

\paragraph{Budget forcing and ambiguity.}
For node $v$, define the aggregate incoming signal
\[
 A_v:=\sum_{u\in N^-(v)}\alpha_{b_u}.
\]
With the additional leakage controlled, the robust wire yields the same two-sided budget response as in the main pacing proof, now with enough slack to absorb the relaxed allocation rule.

\begin{lemma}[Robust Pacing Budget Forcing]
\label{lem:independent-pacing-budget-forcing}
Every $(\sigma,\gamma,\tau)$-approximate pacing equilibrium satisfies
\begin{align}
 A_v+\eta\le\frac\beta2B_v
 &\implies \alpha_{b_v}\ge h>U,
 \label{eq:independent-pacing-force-high}\\
 A_v>B_v
 &\implies \alpha_{b_v}\le L.
 \label{eq:independent-pacing-force-low}
\end{align}
\end{lemma}

\begin{proof}
By \Cref{eq:independent-pacing-target-payment-upper}, the payments on incoming wires sum to at most $A_v$, and by \Cref{cor:independent-pacing-backward-leakage}, outgoing payments sum to less than $\eta$. Under the first hypothesis, bidder $b_v$ therefore spends less than $\beta B_v$, so \Cref{def:independent-pacing}(d) gives $\alpha_{b_v}\ge h>U$.

For the second implication, suppose $A_v>B_v$ and $\alpha_{b_v}>L$. The exact-forward part of \Cref{lem:independent-pacing-geometric-wire} makes the incoming-wire payment equal to $A_v>B_v$, contradicting budget feasibility. Hence $\alpha_{b_v}\le L$.
\end{proof}

\begin{lemma}[Robust Pacing Ambiguity Interval]
\label{lem:independent-pacing-ambiguity-interval}
Consider a unary-gate output $v$ with budget $B$ and predecessor $u$. If $\chi(v)=\bot$, then
\[
 \alpha_{b_u}\in J(B):=\left(\frac\beta2B-\eta,B\right].
\]
\end{lemma}

\begin{proof}
Write $a:=\alpha_{b_u}$. If $a+\eta\le(\beta/2)B$, \Cref{lem:independent-pacing-budget-forcing} gives $\alpha_{b_v}\ge h>U$, contradicting $\chi(v)=\bot$. Hence $a>(\beta/2)B-\eta$. If $a>B$, the same lemma gives $\alpha_{b_v}\le L$, again contradicting $\chi(v)=\bot$. Therefore $a\in((\beta/2)B-\eta,B]$.
\end{proof}

\paragraph{Gate verification and decoding.}
We instantiate the abstract framework with
\[
 z_u:=\alpha_{b_u},
 \qquad
 s_u:=\alpha_{b_u},
 \qquad
 a_0:=L,
 \qquad
 a_1:=U=\frac h2.
\]
By \Cref{lem:independent-pacing-ambiguity-interval}, we may use
\[
 \tau_1(B):=\frac\beta2B-\eta,
 \qquad
 \tau_0(B):=B,
 \qquad
 J(B):=\left(\frac\beta2B-\eta,B\right].
\]
The numerical slack needed for the chosen budgets follows from
\begin{align}
 L+\eta
 &<\frac{11\beta^2h}{4096}
 <\frac{\beta^2h}{64}
 =\frac\beta2B_1
 <\frac{\beta h}{8}
 =\frac\beta2B_2,
 \label{eq:independent-pacing-one-low-bound}\\
 2L+\eta
 &<\frac{19\beta^2h}{4096}
 <\frac{\beta h}{8}
 =\frac\beta2B_G.
 \label{eq:independent-pacing-two-low-bound}
\end{align}
Thus $2a_0<\tau_1(B_G)$ and $\tau_0(B_G)=h/4<a_1$, so $B_G$ satisfies the \NOT/\NOR conditions of \Cref{lem:abstract-gate-simulation}. For each $i\in\{1,2\}$, the first displayed chain gives $a_0<\tau_1(B_i)$, while $\tau_0(B_i)=B_i<a_1$. Finally,
\[
 \frac{\beta h}{32}=B_1
 <\frac{\beta h}{8}-\eta,
\]
because $3\beta h/32-\eta>0$. Hence $J(B_1)\cap J(B_2)=\varnothing$, and every hypothesis of the abstract gate lemma is satisfied.

\begin{proof}[Proof of \Cref{thm:independent-pacing}]
Take the market constructed from a bounded-degree \PureCircuit instance and let $(\alpha,x)$ be any $(\sigma,\gamma,\tau)$-approximate pacing equilibrium. Decode each node by \Cref{eq:independent-pacing-decoder}. A decoded zero has source signal at most $L=a_0$, while a decoded one has source signal at least $U=a_1$. \Cref{lem:independent-pacing-budget-forcing,lem:independent-pacing-ambiguity-interval} provide the abstract thresholds and intervals above, and the gate verification shows that the assigned budgets satisfy \Cref{lem:abstract-gate-simulation}. Applying that lemma to every gate gives a valid \PureCircuit solution.

It remains to check the reduction size and support. Each wire good has exactly two positive valuations, so $\goodsupport=2$. Since every circuit node has total interaction degree at most three, each variable bidder values at most $3K$ goods. Therefore
\[
 \biddersupport\le3K
 =O\left(((1-\gamma)^2(1-\tau))^{-1}\right).
\]
For an $N$-node circuit, the market has $O(NK)$ goods. For fixed $\sigma,\gamma,\tau$, the value of $K$ and all numerical parameters are constant, so the reduction is polynomial time. Together with \Cref{thm:pure-circuit}, this proves the theorem.
\end{proof}

As a consistency check, coupling the independent margins recovers the pacing approximation studied by Chen, Kroer, and Kumar.

\begin{corollary}[Coupled Pacing Approximation]
\label{cor:independent-pacing-coupled}
For every fixed $\delta,\rho\in[0,1)$, computing a $(\delta,\rho)$-approximate pacing equilibrium in the sense of Chen, Kroer, and Kumar is $\PPAD$-hard.
\end{corollary}

\begin{proof}
Set $(\sigma,\gamma,\tau)=(\delta,\rho,\rho)$ in \Cref{thm:independent-pacing}.
\end{proof}

\paragraph{Inverse-Polynomial Endpoint Gaps.}
The explicit parameter dependence further shows that the construction remains polynomial when all three distances to their respective endpoints are inverse-polynomial. For fixed nonnegative integers $c_\sigma,c_\gamma,c_\tau$ and an $N$-node source instance, set
\[
 \sigma_N=1-N^{-c_\sigma},\qquad
 \gamma_N=1-N^{-c_\gamma},\qquad
 \tau_N=1-N^{-c_\tau}.
\]
Then $K=O(N^{2c_\gamma+c_\tau})$ and $Q=4N^{c_\sigma}$. The market has $O(NK)$ goods, and the powers defining the wire thresholds have polynomial bit length, so the construction remains polynomial time and the same decoding proof applies. If $M_N$ denotes the encoding length of the produced instance, then $N\le M_N\le N^D$ for a constant $D$ depending only on the three exponents. Hence every nonconstant endpoint gap remains inverse-polynomial in $M_N$.

\paragraph{Endpoint Cases.}
If either $\gamma=1$ or $\tau=1$, the all-zero multiplier vector together with the zero allocation satisfies the natural endpoint extension of \Cref{def:independent-pacing}. When $\gamma,\tau<1$, the boundary $\sigma=1$ is different: winner eligibility becomes vacuous while complementarity remains active, and the robust wire requires $s=1-\sigma>0$. We therefore make no hardness claim for this remaining boundary case.

\section{Independent Relaxations for Throttling}
\label{app:independent-throttling}

The main throttling theorem uses one parameter both for the spending threshold and for the participation level required by complementarity. This appendix separates those two margins. After a simple rescaling, the same blocker amplifier continues to provide a low--high signal gap, and hardness persists for every fixed positive $\rho,\tau<1$ with constant positive-support degrees.

\begin{definition}[Independent Throttling Approximation]
\label{def:independent-throttling}
For $\rho,\tau\in[0,1]$, a vector $\theta\in[0,1]^n$ is a \emph{$(\rho,\tau)$-approximate throttling equilibrium} if, for every bidder $i$,
\begin{align}
 S_i(\theta)&\le B_i,
 \label{eq:independent-throttling-budget-feasibility}\\
 S_i(\theta)<(1-\rho)B_i
 &\implies \theta_i\ge1-\tau.
 \label{eq:independent-throttling-complementarity}
\end{align}
\end{definition}

Here $\rho$ controls the spending trigger and $\tau$ controls the required participation level. Setting $\rho=\tau=\delta$ recovers \Cref{def:approx-throttling}.

\begin{theorem}[Independent-Relaxation Throttling $\PPAD$-Hardness]
\label{thm:independent-throttling}
For every fixed $\rho,\tau\in(0,1)$, computing a $(\rho,\tau)$-approximate throttling equilibrium is $\PPAD$-hard. Hardness holds with
\[
 \goodsupport(\mathcal I),\ \biddersupport(\mathcal I)=O_{\rho,\tau}(1).
\]
\end{theorem}

Fix $\rho,\tau\in(0,1)$ and set
\[
 \beta:=1-\rho,
 \qquad
 H:=1-\tau,
\]
so $\beta,H\in(0,1)$. Here $\beta$ is the spending-trigger margin and $H$ is the required participation level. We again reduce from bounded-degree \PureCircuit.

\subsection{Construction}
\label{subsec:independent-throttling-construction}

We retain the topology of the main throttling reduction and rescale the decoding threshold, blocker budget, and gate budgets to reflect the two independent margins. Set
\[
 L:=\frac{\beta^2H^2}{256},
 \qquad
 \lambda:=\frac H4,
 \qquad
 p:=\frac L{12},
 \qquad
 C:=\frac{H^2}{8},
\]
and decode node $v$ by
\begin{equation}
\label{eq:independent-throttling-decoder}
 \chi(v):=\dec_{L,H}(\theta_v).
\end{equation}
Because $L/H=\beta^2H/256<1$, the low and high decoding regions are disjoint.

Choose the least positive integer $K$ satisfying
\begin{equation}
\label{eq:independent-throttling-amplifier-size}
 8\left(\frac{1-H}{1-H/4}\right)^K<\beta^2L.
\end{equation}
The ratio lies in $[0,1)$, so such a $K$ exists; when $H=1$, already $K=1$ suffices.

For every circuit node $v$, create one variable bidder $v$. For every source node $u$ of positive out-degree, create blockers $z_{u,1},\ldots,z_{u,K}$, each with budget $C$, and reuse this bundle on every edge leaving $u$. For each blocker $z_{u,q}$, create a control good $c_{u,q}$ on which $z_{u,q}$ bids $2$ and $u$ bids $1$. For every interaction edge $u\to v$, create a signal good $s_{uv}$ on which every blocker $z_{u,q}$ bids $3p$, the target $v$ bids $2p$, and the source $u$ bids $p$. All other bids are zero.

For every source node $u$ of positive out-degree, define
\begin{equation}
\label{eq:independent-throttling-signal-coefficient}
 y_u:=p\theta_u\prod_{q=1}^K(1-\theta_{z_{u,q}}),
\end{equation}
and the two signal scales
\begin{equation}
\label{eq:independent-throttling-signal-scales}
 a_0:=pL(1-H)^K,
 \qquad
 a_1:=pH(1-H/4)^K.
\end{equation}
Finally, set
\begin{equation}
\label{eq:independent-throttling-gate-budgets}
 B_1:=\frac{\beta L^2a_1}{4H},
 \qquad
 B_2:=\frac{La_1}{2}.
\end{equation}
Assign budget $B_1$ to every \NOT and \NOR output. For an \NPURIFY gate, assign its two outputs budgets $B_1$ and $B_2$, respectively. Every blocker keeps budget $C$. This completes the market construction.

\subsection{Proof of Correctness}
\label{subsec:independent-throttling-correctness}

The correctness proof follows the same sequence as before. The rescaled control goods first force the desired blocker response and therefore a separated source signal. Exact variable spending then yields the threshold test and ambiguity interval, and the final parameter check lets us invoke \Cref{lem:abstract-gate-simulation}.

\paragraph{Signal behavior.}
The next lemma packages the two steps that matter for the independent parameters: the control good forces the blocker response, and the product of blocker absences turns that response into a signal gap.

\begin{lemma}[Robust Blocker Amplification]
\label{lem:independent-throttling-amplifier}
Every $(\rho,\tau)$-approximate throttling equilibrium satisfies, for every source node $u$ of positive out-degree and every blocker $z_{u,q}$,
\[
 \theta_u\le L\implies\theta_{z_{u,q}}\ge H,
 \qquad
 \theta_u\ge H\implies\theta_{z_{u,q}}\le\lambda.
\]
Consequently,
\[
 \theta_u\le L\implies y_u\le a_0,
 \qquad
 \theta_u\ge H\implies y_u\ge a_1,
\]
and
\begin{equation}
\label{eq:independent-throttling-master-inequality}
 8Ha_0<\beta^2 L^2a_1.
\end{equation}
\end{lemma}

\begin{proof}
Fix a blocker $z:=z_{u,q}$. As in the main construction, the control good contributes $\theta_z\theta_u$, while at most two outgoing signal goods contribute at most $6p\theta_z$. Hence
\[
 \theta_z\theta_u
 \le S_z(\theta)
 \le \theta_z(\theta_u+6p).
\]
If $\theta_u\le L$, then
\[
 S_z(\theta)
 \le L+6p
 =\frac{3L}{2}
 =\frac{3\beta^2H^2}{512}
 <\frac{\beta H^2}{8}
 =\beta C.
\]
Condition~\eqref{eq:independent-throttling-complementarity} therefore forces $\theta_z\ge H$. If instead $\theta_u\ge H$, budget feasibility gives
\[
 \theta_z H
 \le\theta_z\theta_u
 \le S_z(\theta)
 \le C,
\]
so
\[
 \theta_z\le\frac CH=\frac H8<\frac H4=\lambda.
\]
Substituting these blocker bounds into \Cref{eq:independent-throttling-signal-coefficient} gives the stated low and high signal bounds. Finally, multiplying \Cref{eq:independent-throttling-amplifier-size} by $p H L(1-H/4)^K$ yields
\[
 8Ha_0<\beta^2 L^2a_1.
\]
In particular, $a_0<a_1$ because $\beta^2 L^2<8H$.
\end{proof}

\paragraph{Budget forcing and ambiguity.}
For each variable bidder $v$, let
\[
 A_v:=\sum_{u\in N^-(v)}y_u.
\]
By construction, a variable bidder pays only on incoming signal goods, so her spending admits an exact identity.

\begin{corollary}[Exact Variable Spending]
\label{cor:independent-throttling-exact-variable-spending}
Every variable bidder $v$ satisfies
\begin{equation}
\label{eq:independent-throttling-variable-spending}
 S_v(\theta)=\theta_vA_v.
\end{equation}
\end{corollary}

\begin{proof}
On every incoming signal good $s_{uv}$, conditional on $v$ participating, the target pays $p$ exactly when the source participates and every blocker in the shared bundle is absent. Thus her expected payment on that good is $\theta_v y_u$. Variable bidders pay zero on outgoing signal and control goods. Summing over incoming edges gives \Cref{eq:independent-throttling-variable-spending}.
\end{proof}

\begin{lemma}[Robust Throttling Budget Forcing]
\label{lem:independent-throttling-budget-forcing}
For a variable bidder $v$ with budget $B$, every $(\rho,\tau)$-approximate throttling equilibrium satisfies
\[
 A_v<\beta B\implies\theta_v\ge H,
 \qquad
 A_v>B/L\implies\theta_v\le L.
\]
\end{lemma}

\begin{proof}
Since $S_v(\theta)=\theta_vA_v\le A_v$, the inequality $A_v<\beta B$ triggers \Cref{eq:independent-throttling-complementarity} and forces $\theta_v\ge H$. Conversely, if $A_v>B/L$ and $\theta_v>L$, then $S_v(\theta)=\theta_vA_v>L A_v>B$, contradicting budget feasibility.
\end{proof}

\begin{lemma}[Robust Throttling Ambiguity Interval]
\label{lem:independent-throttling-ambiguity-interval}
Consider a unary-gate output $v$ with budget $B$ and predecessor $u$. If $\chi(v)=\bot$, then
\[
 y_u\in J(B):=\left(\frac{\beta B}{H},\frac{B}{L}\right).
\]
\end{lemma}

\begin{proof}
Write $\theta:=\theta_v$ and $y:=y_u$. By \Cref{eq:independent-throttling-decoder}, $\chi(v)=\bot$ means $L<\theta<H$. Since $\theta<H$, the contrapositive of \eqref{eq:independent-throttling-complementarity} gives $\theta y\ge\beta B$, and therefore $y>\beta B/H$. Budget feasibility gives $\theta y\le B$; because $\theta>L$, we also have $y<B/L$. This proves the stated interval.
\end{proof}

\paragraph{Gate verification and decoding.}
We instantiate the abstract framework with
\[
 z_u:=\theta_u,
 \qquad
 s_u:=y_u.
\]
By \Cref{lem:independent-throttling-amplifier}, decoded zeros and ones send signals at most $a_0$ and at least $a_1$, respectively. By \Cref{lem:independent-throttling-budget-forcing,lem:independent-throttling-ambiguity-interval}, we may use
\[
 \tau_1(B):=\beta B,
 \qquad
 \tau_0(B):=\frac BL,
 \qquad
 J(B):=\left(\frac{\beta B}{H},\frac BL\right).
\]
The master inequality gives
\[
 2a_0<\frac{\beta^2 L^2a_1}{4H}=\beta B_1.
\]
Moreover,
\[
 \frac{B_1}{B_2}=\frac{\beta L}{2H}=\frac{\beta^3H}{512}<1,
\]
so $B_1<B_2$, and
\[
 \frac{B_1}{L}=\frac{\beta La_1}{4H}<a_1,
 \qquad
 \frac{B_2}{L}=\frac{a_1}{2}<a_1.
\]
Thus $B_1$ satisfies the \NOT/\NOR conditions, and both $B_1$ and $B_2$ satisfy the Boolean-input conditions for \NPURIFY. Finally,
\[
 \frac{B_1}{L}=\frac{\beta La_1}{4H}
 <\frac{\beta La_1}{2H}=\frac{\beta B_2}{H},
\]
so $J(B_1)\cap J(B_2)=\varnothing$. Every hypothesis of \Cref{lem:abstract-gate-simulation} is therefore satisfied.

\begin{proof}[Proof of \Cref{thm:independent-throttling}]
Take the market constructed from a bounded-degree \PureCircuit instance and let $\theta$ be any $(\rho,\tau)$-approximate throttling equilibrium. Decode each node by \Cref{eq:independent-throttling-decoder}. \Cref{lem:independent-throttling-amplifier} supplies the abstract signal gap, while \Cref{lem:independent-throttling-budget-forcing,lem:independent-throttling-ambiguity-interval} supply the required thresholds and intervals. The gate verification above therefore lets us apply \Cref{lem:abstract-gate-simulation} to every gate. For \NPURIFY, the two outputs receive the same predecessor signal because the construction reuses the source's blocker bundle. Hence the decoded assignment is a valid \PureCircuit solution.

For an $N$-node source circuit, the market has $O(NK)$ bidders and goods. Every control good has two positive bids, every signal good has $K+2$, every variable bidder bids positively on at most $K+3$ goods, and every blocker on at most three. Hence
\[
 \goodsupport\le K+2,
 \qquad
 \biddersupport\le K+3.
\]
For fixed $\rho,\tau$, the value of $K$ and all numerical parameters are constant, so the reduction is polynomial time. Together with \Cref{thm:pure-circuit}, this proves the theorem.
\end{proof}

\paragraph{Inverse-Polynomial Endpoint Gaps.}
The same explicit estimates keep the reduction polynomial when the two endpoint gaps are inverse-polynomial. For fixed positive integers $c_\rho,c_\tau$ and an $N$-node source instance, set
\[
 \rho_N=1-N^{-c_\rho},
 \qquad
 \tau_N=1-N^{-c_\tau}.
\]
Then $\beta=N^{-c_\rho}$ and $H=N^{-c_\tau}$. Since
\[
 \frac{1-H}{1-H/4}\le e^{-3H/4},
\]
it is enough to take
\[
 K>\frac{4}{3H}\left(\log 2048+4\log\frac1\beta+2\log\frac1H\right).
\]
For $N\ge2$, the explicit choice
\[
 K_N:=\left\lceil
 8N^{c_\tau}\left(2+(4c_\rho+2c_\tau)\lceil\log_2N\rceil\right)
 \right\rceil
\]
dominates this sufficient bound and therefore satisfies \Cref{eq:independent-throttling-amplifier-size}. Thus $K_N=O(N^{c_\tau}(1+(c_\rho+c_\tau)\log N))$. All primitive bids and budgets have polynomial bit length, and the powers appearing in $a_1$ and the gate budgets have $O(K_N\log N)$ bits. Hence the produced market has polynomial binary encoding length, every nonconstant endpoint gap remains inverse-polynomial in that length, and the same decoding proof gives $\PPAD$-hardness throughout this regime.

\paragraph{Endpoint Cases.}
If either $\rho=1$ or $\tau=1$, the all-zero participation vector satisfies \Cref{def:independent-throttling}.

\end{document}